\documentclass{patmorin}
\usepackage{pat,graphicx,amsmath}
\usepackage[mathlines]{lineno}
\newcommand{\depth}{\mathrm{depth}}

\title{\MakeUppercase{Point Location in Disconnected Planar Subdivisions}}

\author{Prosenjit~Bose, 
        Luc~Devroye,
	Karim~Dou\"{\i}eb, 
	Vida~Dujmovi\'c, 
	James~King, and 
	Pat~Morin}

\begin{document}
\maketitle

\begin{abstract}
Let $G$ be a (possibly disconnected) planar subdivision and let $D$
be a probability measure over $\R^2$.  The current paper shows how to
preprocess $(G,D)$ into an $O(n)$ size data structure that can answer
planar point location queries over $G$.  The expected query time of this
data structure, for a query point drawn according to $D$, is $O(H+1)$,
where $H$ is a lower bound on the expected query time of any linear
decision tree for point location in $G$.  This extends the results
of Collette \etal\ (2008, 2009) from connected planar subdivisions to
disconnected planar subdivisions.  A version of this structure, when
combined with existing results on succinct point location, provides a
succinct distribution-sensitive point location structure.
\end{abstract}

\section{Introduction}

Planar point location is the classic search problem in computational
geometry.  The problem asks us to preprocess a planar subdivision
$G$ so that we can quickly test, for any query point $p$, which
face of $G$ contains $p$.  Optimal, $O(n)$ space, $O(\log n)$ query
time structures for the point location problem have been known for
over 25 years \cite{egs86,k83,m90,st86}, the precise constants
achievable in the query time are well-understood \cite{as98},
several results exist for distribution-sensitive query times
\cite{acmr00,amm00,amm01a,amm01b,ammw07,cdilm08,cdilm09,i01,i04}, and
sublogarithmic query time data structures exist for transdichotomous
models of computation \cite{c06,cp09,p06}.

The most recent work in the distribution-sensitive setting is by
Collette \etal\ \cite{cdilm08} who give an $O(n)$ space data structure
that preprocesses a connected planar subdivision $G$ and a probability
measure $D$ over $\R^2$ such that a point location query in $G$ can be
answered in $O(H+1)$ expected time.  Here $H$ is a lower-bound on the
expected time required by any linear decision tree for answering queries
on $G$ that are drawn according to $D$.  The expected number of point-line
comparisons needed to answer a query using their data structure is $H +
O(H^{2/3}+1)$.  Their work, which generalizes (and uses) a similar result
for triangulations \cite{ammw07}, leaves open the problem of what to do
when $G$ is disconnected. Disconnected planar subdivisions occur quite
frequently in areas like geographic information systems and cartography,
where disconnected regions occur naturally. (See \figref{falklands}
for example).

\begin{figure}
  \begin{center}\includegraphics[width=5in]{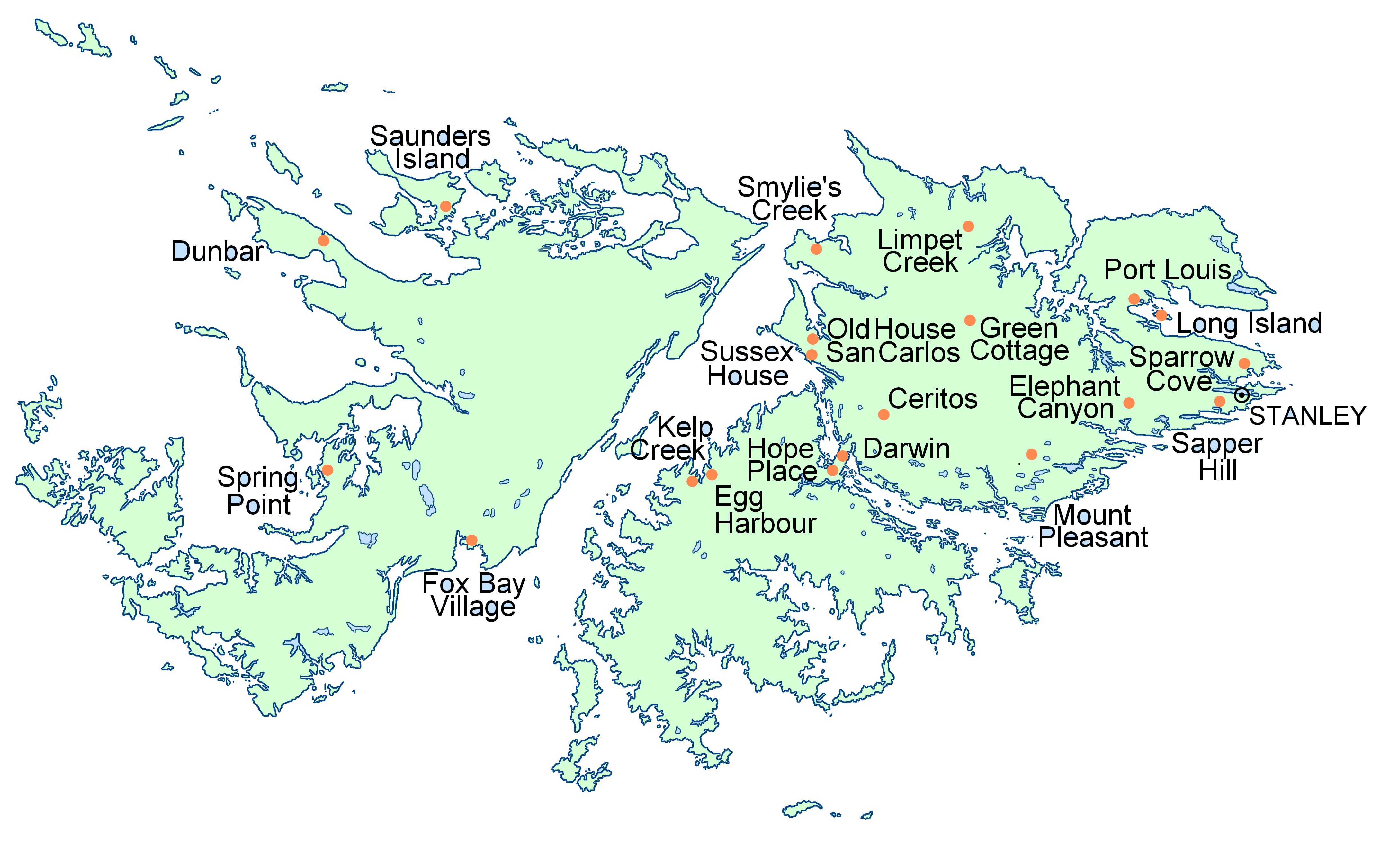}\end{center}
  \caption{A disconnected planar subdivision that occurs in the context of
cartography.}
  \figlabel{falklands}
\end{figure}

In the current paper we show that, for a (possibly disconnected)
planar subdivision $G$, a very different approach can be used to obtain
an expected query time of $O(H+1)$. Essentially, the problem can be
solved by building a $o(n)$-sized data structure for answering the
easy-to-answer queries efficiently and passing all other (hard-to-answer)
queries on to any of the classic $O(n)$ space $O(\log n)$ query time
data structures for planar point location. As a corollary, we obtain
a succinct distribution-sensitive data structure for point location in
(possibly-disconnected) subdivisions.  This data structure stores only
a permutation of the vertices of the subdivision plus an additional
$o(n)$ bits.

\section{Preliminaries}

Throughout this paper, we assume an underlying probability measure
$D$ over $\R^2$.  All expectations and probabilities are (implicitly)
with respect to $D$.  For any subset $X\subseteq\R^2$, $\Pr(X)$ refers
to $D(X)$.  We use the notation $D_{|X}$ to denote the distribution $D$
conditioned on $X$, i.e., $D_{|X}(Y)=\Pr(Y\mid X)=\Pr(X\cap Y)/\Pr(X)$
for all $Y\subseteq\R^2$.  If $\Delta$ is a partition of $\R^2$, then
the \emph{entropy} of $\Delta$, denoted $H(\Delta)$ is
\[
    H(\Delta) = \sum_{t\in \Delta} \Pr(t)\log(1/\Pr(t)) \enspace .
\]
The probability measure $D$ is used as an input to our algorithms.
We assume that the algorithm has access to $D$ through two oracles.
The first oracle allows, for any triangle $t$, to determine $\Pr(t)$
in constant time.  The second oracle, for any triangle $t$, allows the
algorithm to draw a point $p$ according to $D_{|t}$ in constant time.

A \emph{linear decision tree} for point location over $G$ is a rooted
binary tree in which each internal node $v$ is labelled with a linear
inequality $a_vx + b_vy + c_v > 0$, and each leaf $\ell$ is labelled
with a face of $G$.  A query point $p=(x,y)$ follows a root-to-leaf path,
proceeding to the left child of $v$ if it satisfies the inequality and the
right child of $v$ if it does not.  A linear decision tree is \emph{for
point location in $G$} if, for every $p\in\R^2$ the path for $p$ ends at a
leaf labelled with the face of $G$ that contains $p$. In the case where
$p$ lies on an edge or vertex of $G$, the label can be any of the faces
of $G$ incident on that edge or vertex.  The \emph{(expected) cost} of a
linear decision tree is the expected depth of the leaf reached when $p$
is drawn according to the probability measure $D$.

\section{The Data Structure}
\seclabel{data-structure}

In this section we describe our data structure for point location in
disconnected planar subdivisions.  The first tool we use is simplicial
partitions, from the field of geometric range searching:

\begin{thm}[Matou\v{s}ek 1992]\thmlabel{point-partition}
There exists a universal constant $c$ such that, for any set $S$ of $m$
points in $\R^2$ and any $r\in\{1,\ldots,m\}$, there exists a sequence
$\langle \Delta_1,\ldots,\Delta_r\rangle$ of closed triangles such that
  \begin{enumerate}
    \item $S\subseteq \bigcup_{i=1}^r \Delta_i$,
  
    \item $\left|\Delta_i \cap S\setminus
    \left(\bigcup_{j=1}^{i-1}\Delta_j\right)\right| \le 2m/r$, and
  
    \item For any line $\ell$, there are at most $cr^{1/2}$ elements of
  $\{\Delta_1,\ldots,\Delta_r\}$ whose interiors intersect $\ell$.
  \end{enumerate}
  The sequence of triangles $\Delta_1,\ldots,\Delta_r$ can be computed
  in $O(m)$ time.
\end{thm}

Note that Part~2 of \thmref{point-partition} is not in the original
statement of the theorem, but follows from Matou\v{s}ek's construction
of $\Delta_1,\ldots,\Delta_r$ \cite{m92}.
Restating \thmref{point-partition} in terms of probability distributions,
we have:

\begin{thm}\thmlabel{prob-partition}
There exists a universal
constant $c$ such that, for any probability measure $D$ over $\R^2$ and any 
integer $r\ge 1$, there exists a sequence
$\langle \Delta_1,\ldots,\Delta_r\rangle$ of closed triangles such that
  \begin{enumerate}
    \item $\Pr\left\{\bigcup_{i=1}^r \Delta_i\right\} = 1$,
  
    \item $\Pr\left\{\Delta_i \setminus
    \left(\bigcup_{j=1}^{i-1}\Delta_j\right)\right\} \le 3/r$, and
  
    \item For any line $\ell$, there are at most $cr^{1/2}$ elements of
    $\{\Delta_1,\ldots,\Delta_r\}$ whose interiors intersect $\ell$.
  \end{enumerate}
  The sequence $\Delta_1,\ldots,\Delta_r$ of triangles can be computed
  in $O(r^3\log r)$ time.
\end{thm}

\def\isdef{\buildrel {\rm def} \over =}
\def\PROB{\Pr}

\begin{proof}
Assume that $r \ge 2$, otherwise the theorem is trivial.
We will draw an i.i.d.\ sample of $m = \lceil 256r^3\ln r \rceil$
points from $D$ to form a set $S$.  We use the algorithm from Theorem
1 to build a sequence $\langle \Delta_1, \ldots, \Delta_r \rangle$
of triangles satisfying the conditions of Theorem 1.  If necessary we
replace $\Delta_r$ with a triangle that contains the support of $D$
to ensure condition (1) of this theorem is satisfied.  Condition (3) of
this theorem is the same as condition (3) of Theorem 1 and is therefore
trivially satisfied, though it may be necessary to add 1 to the constant
$c$ due to the replacement of $\Delta_r$.

We will prove that, with probability at least $1/2$, the sequence $\langle
\Delta_1, \ldots, \Delta_r \rangle$ also satisfies condition~(2) of
this theorem.  Our oracles allow us to check in constant time whether
this condition is satisfied; we repeat the process until we obtain
a partition that does.  The runtime for this algorithm will then be
geometrically distributed with constant expectation for any constant $r$.

To denote the incremental differences between the triangles we use
\[
  \Delta_i^* = \Delta_i \setminus \bigcup_{j=1}^{i-1} \Delta_j ~.
\]
We will
use $D_m(A)$ to denote the empirical measure of a set $A$:
\[
  D_m (A) \isdef {{|S\cap A|} \over m}
~.
\]  
By condition~(2) of \thmref{point-partition},
we have
\[
  \sup_{1 \le i \le r} D_m(\Delta_i^*) \le { 2 \over r }~.
\]
Now,
\begin{eqnarray*}
\PROB \left\{ \sup_{1 \le i \le r} D ( \Delta^*_i ) > {3 \over r} \right\}
&=& \PROB \left\{ \cup_{1 \le i \le r} \left[ D ( \Delta^*_i ) - D_m(\Delta_i^*)  >  {3 \over r} - D_m(\Delta_i^*) \right] \right\} \\
&\le& \PROB \left\{ \cup_{1 \le i \le r} \left[ D ( \Delta^*_i ) - D_m(\Delta_i^*)  > {3 \over r} - {2 \over r}  \right] \right\} \\
&=& \PROB \left\{ \sup_{1 \le i \le r} \left(  D ( \Delta^*_i ) - D_m(\Delta_i^*) \right) > {1 \over r} \right\} \\
&\le& \PROB \left\{ \sup_{A \in {\cal A}} \left( D ( A ) - D_m ( A ) \right)  > {1 \over r} \right\} \\
\end{eqnarray*}
where $\cal A$ are sets formed by taking a closed triangle and subtracting
at most $r-1$  closed triangles from it.
The class $\cal A$ for $r=1$ is the class of all triangles.
It has Vapnik-Chervonenkis dimension at most 7. 
By Sauer's  lemma 
\cite{s72}\cite[Pages~28--29]{dl01},
the number of subsets of an $m$-point set that can be obtained by intersections
with sets from $\cal A$ does not exceed $(m+1)^7$.
Assume now general $r$. Then the number of subsets of an $n$-point set that can be obtained by intersections
with sets from $\cal A$ does not exceed $(m+1)^{7r}$,
by a simple combinatorial argument.
Then, by a version of the Vapnik-Chervonenkis Inequality 
\cite{vc71} shown by Devroye \cite{d82},
\[
\PROB \left\{ \sup_{A \in {\cal A}} \left| D ( A ) - D_m ( A ) \right|  \ge t \right\}
\le 4 e^{4t+4t^2} \left( m^2 + 1 \right)^{7r}  e^{-2mt^2} \enspace , 
\]
for any $t > 0$.  Thus, 
\begin{eqnarray*}
\PROB \left\{ \sup_{1 \le i \le r} D ( \Delta^*_i )  > {3 \over r} \right\}
&\le& 4 e^{4/r + 4/r^2} \left( m^2 + 1 \right)^{7r}  e^{-2m/r^2} \\
&\le& 2^{2+7r}  e^{8}  m^{14r}  e^{-2m/r^2} \\
&\le& \exp \left( 31r\ln m - \frac{2m}{r^2} \right)  ~.
\end{eqnarray*}
Since we have $m=\lceil 256r^3\ln r \rceil$, this upper bound is less than $1/2$, as desired.  This concludes the proof.
\end{proof}

Assume, without loss of generality, that all vertices of $G$ and the support
of $D$ are contained in the unit square $[0,1]^2$.  This can easily be
justified by scaling and translation, so that $G$ is contained in $[0,1]^2$, and performing 4 point-line comparisons
to check that the the query point is in $[0,1]^2$ before using the data structure to answer a query.

We use \thmref{prob-partition} to recursively construct a \emph{partition
tree} $T$.  Let $\alpha > 0$ be a constant that will be specified below.
Refer to \figref{simp-part}. At the root of $T$, we find the sequence of
triangles $\Delta=\langle\Delta_1,\ldots,\Delta_r\rangle$ and construct
the arrangement of triangles in $\Delta$.  

\begin{figure}
  \begin{center}
    \begin{tabular}{ccc}
      \includegraphics[scale=0.77]{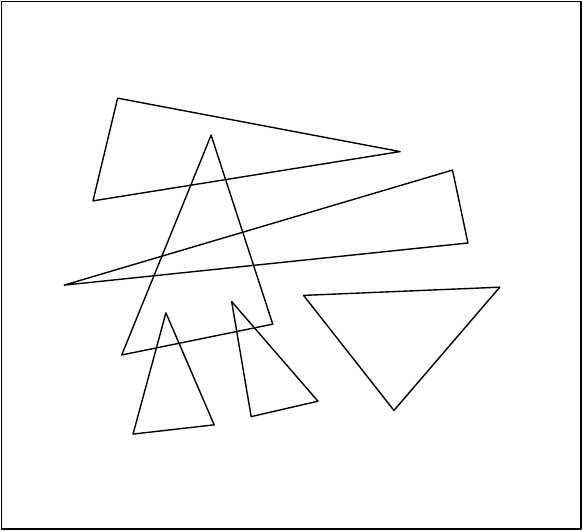} &
      \includegraphics[scale=0.77]{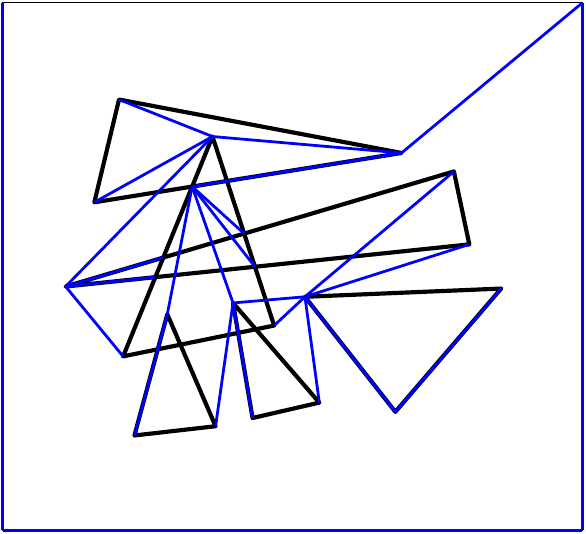} &
      \includegraphics[scale=0.77]{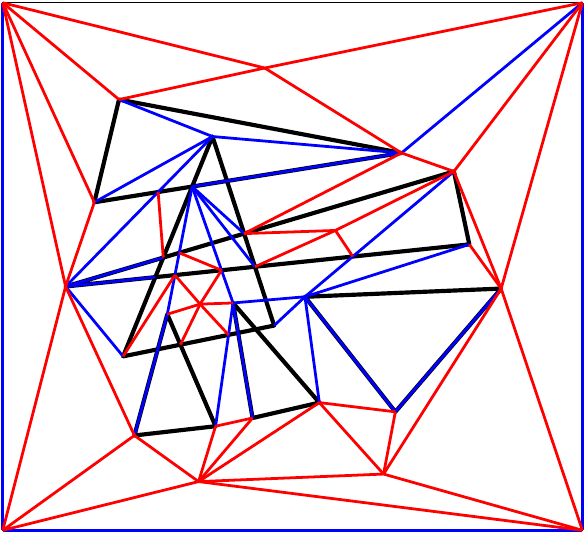} \\
      (a) & (b) & (c) \\
    \end{tabular}
  \end{center}
  \caption{The triangles of a simplicial partition (a) form an
    arrangement of triangles to which (b) a spanning tree is added, and
    (c) the faces of the resulting connected subdivision are (Steiner)
    triangulated to form a Steiner triangulation $A$.}
  \figlabel{simp-part}
\end{figure}

Next, we describe how to triangulate this arrangement while maintaining
the properties of \thmref{prob-partition}. Let $\mathcal{V}$ be set of
$3r+4$ points that make up the vertices of the triangles in $\Delta$ plus
the vertices of a square $\Box$ that contains all triangles in $\Delta$.
A classic result of Haussler and Welzl \cite{hw87} proves that $V$ has
a spanning tree $T(V)$ such that any line crosses $O(r^{1/2})$ edges of
$T(V)$, and this spanning tree can be constructed efficiently \cite{cw89}.
(See \figref{simp-part}.b.)

Consider the line segment arrangement $L$ consisting of the union of
the edges in $T(V)$, the triangles in $\Delta$, and the edges of $\Box$.
Note that any line $\ell$  intersects $O(r^{1/2})$ edges of the
arrangement $L$; $O(r^{1/2})$ of these intersections are generated by
edges corresponding to edges of $T(V)$ and $O(r^{1/2})$ are generated
by edges of triangles in $\Delta$.  What remains is to show how to
triangulate the faces of $L$ without introducing too many crossings.

By construction, each face $F$ of $L$, except the outer face, is a
(weakly) simple polygon having $O(r)$ vertices and edges on its boundary.
By a result of Hershberger and Suri \cite{hs95}, there exists a Steiner
triangulation, $A(F)$, of $F$ using $O(r)$ vertices such that any chord
of $F$ intersects $O(\log r)$ edges of $A(F)$. We therefore triangulate
the arrangement $L$ by triangulating each of its faces in this way.
This gives a Steiner triangulation $A$ of $L$ in which any line intersects
$O(r^{1/2}\log r)$ edges of $A$.
(See \figref{simp-part}.c.)

Next, each face $F$ of $A$ becomes a child of the root of $T$.  If the
interior of $F$ is contained in a single face of $G$ then we call $F$ a
\emph{terminal leaf} and label $F$ with the face of $G$ that contains it.
If the current depth of recursion is greater than $\lfloor\alpha \log_{r}
n\rfloor$ then $F$ becomes a \emph{non-terminal leaf} of $T$.  Otherwise
($F$ intersects two or more faces of $G$ and its depth is small),
we recursively apply the same procedure on the distribution $D_{|F}$
to obtain a partition tree that becomes a child of the root.

This construction defines a tree $T=T(G,D)$ in which
each node has $O(r^2)$ children and whose height is at most
$\alpha\log_{r} n$.  The number of nodes of $T$ at level $i$ is most
$(O(r^2))^i=O(r^{i(2+\epsilon)})$ and therefore the total number of nodes
in $T$ is $(O(r^2))^{\alpha \log_{r} n+1} = O(n^{2\alpha+\epsilon})$, where
$\epsilon>0$ is a decreasing function of $r$.  Note that, for $\alpha <
1/2$ and sufficiently large $r$, the size of $T$ is $o(n)$.

In addition to the tree $T$ we construct a backup data structure $T'$ that
can answer point location queries in $G$ in $O(\log n)$ worst-case time.
To answer a query, $T$ and $T'$ are used as follows:  We search top-down
in $T$ for the query point. If this search ends at a terminal leaf
$F$ of $T$ then we report the label at $F$ and the query is complete.
Otherwise we use $T'$ to answer the query in $O(\log n)$ time.

\section{Analysis}
\seclabel{analysis}

Collette \etal\ \cite{cdilm08,cdilm09} show that, up to a lower-order
term, the expected number of comparisons performed by the optimal
decision tree for point location in $G$ is equal to the entropy of the
minimum-entropy Steiner triangulation of $G$.

\begin{thm}[Collette \etal\ 2008]\thmlabel{triangulation}
Let $G$ be a planar subdivision and let $D$ be a probability measure
over $\R^2$.  Let $T^*$ be a  minimum-entropy Steiner triangulation of
$G$ and let $H^*$ be the entropy of $T^*$.  Then any linear decision tree
for point location in $G$ has expected cost at least $H^*-O(\log H^*)$.
\end{thm}

Thus, our goal is to prove that our query time approximates the entropy
of the minimum entropy Steiner triangulation of $G$.  We begin by showing
that the partition tree $T$ has small \emph{visiting number} \cite{hw87}.

\begin{lem}\lemlabel{crossing}
  Let $\epsilon > 0$, and let $T$ be the partition tree defined
  in \secref{data-structure} using a value $r$ such that $r>(c\log
  r)^{1/\epsilon}$ for some (sufficiently large) constant $c$.  Then the
  number of nodes of $T$ whose depth is at most $i$ that are intersected
  by any line $\ell$ is $O(r^{i(1/2+\epsilon)})$.
\end{lem}

\begin{proof}
  Recall that each node of $T$ corresponds to a triangle and $T$ has
  the property that the number of children of any node intersected by
  any particular line $\ell$ is $O(r^{1/2}\log r)$. Therefore, the number
  of nodes $x(i)$ of $T$ at level $i$ that intersect $\ell$ is given by
  the recurrence
  \[
     x(i)\le 
      \left\{
       \begin{array}{ll}
          1 & \mbox{for $i=0$} \\
          (cr^{1/2}\log r)\cdot x(i-1) & \mbox{for $i>0$}
       \end{array}
      \right.
\]
which resolves to $(cr^{1/2}\log r)^i= O(r^{i(1/2+\epsilon)})$ for 
$r>(c\log r)^{1/\epsilon}$.
\end{proof}

An \emph{$i$-set} of a rooted tree $T$ is a set of vertices in $T$ all
of which are at distance at most $i$ from the root of $T$ and in which
no vertex in the set is the ancestor of any other vertex in the set.
Note that if $T$ is a partition tree defined in \secref{data-structure}
then an $i$-set of $T$ is a set of disjoint triangles.  We say that
a set of regions $X=\{X_1,\ldots,X_m\}$, $X_i\subseteq\R^2$, is in
\emph{$k$-general position} if there is no line that intersects $k$
or more elements of $X$.

\begin{lem}\lemlabel{independent}
  Let $\epsilon > 0$, let $T$ be the partition tree defined in
  \secref{data-structure} using a value $r>(c\log r)^{1/\epsilon}$ for
  some (sufficiently large) constant $c$, and let $V$ be an $i$-set of
  $T$.  Then $V$ contains a subset $V'\subseteq V$ that is in $k$-general
  position and has size $\Omega(|V|/r^{i(1/2+\epsilon + 4/k)})$.
\end{lem}

\begin{proof}
  We will prove the lemma using the probabilistic method \cite{as08}.
  Let $V'$ be a Bernoulli sample of $V$ where each element is selected
  independently with probability $p=r^{-i(1/2+\epsilon+\delta)}$, where $\delta$ is a constant with $\delta > 4/k$.
  We will show that
  \[
     \Pr\left\{
        \mbox{$V'$ is in $k$-general position 
          and $|V'|=\Omega(|V|/r^{i(1/2+\epsilon+\delta)})$}
      \right\} > 0 \enspace ,
  \]
  thus proving the existence of a set $V'$ satisfying the conditions of
  the lemma.

  Consider any line $\ell$. By \lemref{crossing}, $\ell$
  intersects at most $cr^{i(1/2+\epsilon)}$ elements of $V$ for some
  constant $c$. The probability that $\ell$ intersects $k$ or more
  elements of $V'$ is therefore no more than
  \[
    \binom{cr^{i(1/2+\epsilon)}}{k}\cdot p^k
    \le (cr^{i(1/2+\epsilon)}p)^k
    = c^kr^{ki(1/2+\epsilon)-ki(1/2+\epsilon+\delta)}
    = c^kr^{-ki\delta}
  \]
  The nodes in $V$ define a \emph{test set} $L$ of $O(|V|^2)=O(r^{i(4+\epsilon)})$
  lines such that $V'$ is in $k$-general position if and only if no line
  in $L$ intersects $k$ or more elements of $V'$.   The probability that
  any line in $L$ intersects more than $k$ elements of $V'$ is therefore
  at most $O(r^{i(4+\epsilon)}c^kr^{-ki\delta})=O(c^kr^{i(4+\epsilon-k\delta)})=o(1)$ for any
  constant $\delta > 4/k+\epsilon$.

  The above argument shows that the nodes in $V'$ are quite likely
  to be in $k$-general position.  To see that $V'$ is sufficiently
  large, we simply observe that $|V'|$ is a $\mathrm{binomal}(|V|,p)$
  random variable and therefore has median value at least
  $\lfloor{p|V|}\rfloor=\Omega(|V|/r^{i(1/2+\epsilon + \delta)})$.
  In particular, $\Pr\{|V'|\ge \lfloor{p|V|}\rfloor\}\ge 1/2$.  Therefore,
  \[
     \Pr\left\{
        \mbox{$V'$ is in $k$-general position 
          and $|V'|=\Omega(|V|/r^{i(1/2+\epsilon+\delta)})$}
      \right\} \ge 1- (o(1) + 1/2) > 0 \enspace .
  \]
  Setting $\delta$ sufficiently close to (but larger than) $4/k+\epsilon$
  completes the proof.
\end{proof}

We are now ready to show that the search time in our data structure
is a lower bound on the entropy of any Steiner triangulation of $G$.
Recall that, by \thmref{triangulation}, the entropy of a minimum entropy
Steiner triangulation of $G$ is a lower bound on the expected cost of
any linear decision tree for point location in $G$.

\begin{lem}\lemlabel{lower-bound}
  Let $T$ be the partition tree defined in \secref{data-structure},
  let $L$ denote the set of leaves of $T$, and let $H^*=H(\Delta^*)$ be
  the entropy of a Steiner triangulation $\Delta^*$ of $G$.  Then $H^*
  = \Omega(H(L)-1)$
\end{lem}

\begin{proof}
  This proof mixes the ideas from the proofs of Lemma~3 by Dujmovi\'c
  \etal\ \cite{dhm09} and Lemma~4 by Collette \etal\ \cite{cdilm08}.

  Let $T'$ be the tree obtained from $T$ by removing all terminal leaves,
  and let $L'$ denote the set of leaves of $T'$.  Note that $L'$ is a Steiner
  triangulation of $G$ and that 
  \[  
     H(L') = H(L) - O(\log r) = H(L) - O(1)
  \]
  since each triangle in $L'$ is partitioned in $O(r^2)$ triangles in $L$. 

  Partition $L'$ into groups $G_1,G_2,\ldots$, where $G_i$
  contains all leaves $v$ such that $1/2^{i-1} \ge \Pr(v) \ge
  1/2^{i}$.  Further partition each group $G_i$ into subgroups
  $G_{i,1},\ldots,G_{i,t_i}$ with the property that each group $G_{i,j}$
  with $j\in\{1,\ldots,t_i-1\}$ is in $k$-general position and has size
  at least $2^{\gamma i}$ for some constant $\gamma > 0$. Furthermore,
  the final group, $G_{i,t_i}$ has size at most $O(2^{\beta i})$, for some
  constant $\beta < 1$.  This partitioning is accomplished by repeatedly
  applying \lemref{independent} to remove a subset $G_{i,j}\subseteq
  G_{i}$ that is in $k$-general position and has size $2^{\gamma i}$,
  stopping the process once the size of $G_i$ drops below $2^{\beta
  i}$. This works provided that we choose $\beta$, $k$, and $r$ so
  that $\beta > ((\log r)/(\log r - 1))(1/2+\epsilon+4/k)$ and set
  $\gamma=\beta - ((\log r)/(\log r - 1))(1/2+\epsilon+4/k)$.

  Now, consider any Steiner triangulation $\Delta^*$ of $G$ and let
  $t$ be a triangle in $\Delta^*$.  Note that $t$ cannot contain any
  triangle in $L'$ since each element in $L'$ is non-terminal in $T$
  and therefore its interior intersects at least two faces of $G$.
  Therefore, any subgroup $G_{i,j}$ intersected by $t$ must intersect one
  of $t$'s three edges. Since each $G_{i,j}$ is in $k$-general position,
  this means that $t$ intersects at most $3k$ elements of $G_{i,j}$.
  It follows \cite[Lemma~3]{cdilm09} that
  \[
    H^* \ge H(L') 
       - H(\{\cup G_{i,j}:i\in\N,\, j \in\{1,\ldots,t_{i,j}\}) 
       - O(1) \enspace .
  \]
  Thus, all that remains is to upper-bound the contribution of $\bar{H}=H(\{\cup G_{i,j}:i\in\N,\, j \in\{1,\ldots,t_{i,j}\})$.
  \begin{eqnarray*}
    \bar{H} &= & H(\{\cup G_{i,j}:i\in\N,\, j \in\{1,\ldots,t_{i,j}\}) \\
     & = & \sum_{i=1}^\infty \sum_{j=1}^{t_{i}} 
         \Pr(\cup G_{i,j})\log(1/\Pr(\cup G_{i,j})) \\
   & = & \sum_{i=1}^\infty
        \left( 
          \sum_{j=1}^{t_{i}-1} 
             \Pr(\cup G_{i,j})\log(1/\Pr(\cup G_{i,j}) 
             + \Pr(\cup G_{i,t_i})\log(1/\Pr(\cup G_{i,t_i}))
        \right) \\
   & \le & \sum_{i=1}^\infty
        \left( 
          \sum_{j=1}^{t_{i}-1} 
             \Pr(\cup G_{i,j})\log(2^{i-\alpha i})
             + i 2^{\beta i - i + 1}
        \right) \\
    & \le & (1-\alpha)H(L') + O(1) \enspace .
  \end{eqnarray*}
  Thus, we have 
  \[  
     H^* \ge H(L') - \bar{H} -O(1) \ge \alpha H(L') - O(1) 
         \ge \alpha H(L) - O(1) = \Omega(H(L) - 1) 
  \]
  as required.
\end{proof}

\begin{thm}
  Let $G$ be a (possibly disconnected) planar subdivision of size $n$
  and let $D$ be a probability measure over $\R^2$.  There exists a data
  structure $T$ that, given $G$ and $D$, can be constructed in $O(n)$
  time, has $O(n)$ size, and can answer point location queries in $G$
  in $O(H^*)$ expected time, where $H^*$ is the expected time to answer
  point location queries in $G$ using any linear decision tree.
\end{thm}

\begin{proof}
  The data structure is, of course, the partition tree $T$ of
  \secref{data-structure} and some backup structure that can answer
  queries in $O(\log n)$ worst case time in case a query reaches a
  non-terminal leaf of $T$.  The expected time answer queries in $T$ is
  \[
     \sum_{t\in L} \Pr(t)O(\depth_T(t)) = \sum_{t\in L}\Pr(t)O(\log(1/\Pr(t))) = O(H(L)) \enspace .
  \]
  On the other hand, by \lemref{lower-bound} and \thmref{triangulation},
  the expected time required by any linear decision tree for answering
  queries in $G$ is
  \[
      H^* = \Omega(H(L) - 1) \enspace ,
  \]
  which completes the proof.
\end{proof}

We finish by observing that the tree $T$ in \secref{data-structure}
has sublinear size. Indeed, for any constant $0 \le d \le 1$, we can
construct a tree $T$ of size $O(n^d)$ that satisfies the conditions
of \lemref{lower-bound}.  Thus, we can think of $T$ as a sublinear
sized filter that can take any point location structure with $O(\log n)$
worst-case query time and make it into a distribution-sensitive data
structure.  In particular, one can combine $T$ with the succinct point
location structure of Bose \etal\ \cite[Theorem~2]{bchmm09}, to obtain the
following result:

\begin{thm}
  Let $G$ be a (possibly disconnected) planar subdivision of size $n$
  and let $D$ be a probability measure over $\R^2$.  There exists a data
  structure $T$ that, given $G$ and $D$, can be constructed in $O(n)$
  time and can answer point location queries in $G$ in $O(H^*)$ expected
  time, where $H^*$ is the expected time to answer point location queries
  in $G$ using any linear decision tree.  This structure is represented
  as a permutation of the vertices of $G$ and an additional $o(n)$ bits.
\end{thm}

\bibliographystyle{plain}
\bibliography{entropy3}
\end{document}